\documentclass[runningheads]{llncs}

\usepackage{graphicx}
\usepackage{tikz}
\usepackage{amsfonts}
\usepackage{float}
\usepackage{xcolor}
\usepackage{amsmath}
\usepackage{hyperref}

\newcommand{\no}[1]{}
\newcommand{\zero}{\texttt{0}}
\newcommand{\one}{\texttt{1}}
\newcommand{\two}{\texttt{2}}

\newcommand{\f}{\mathcal{F}}
\newcommand{\syma}{\texttt{a}}
\newcommand{\symb}{\texttt{b}}
\newcommand{\symc}{\texttt{c}}
\newcommand{\symd}{\texttt{d}}

\usetikzlibrary{positioning, arrows}

\begin{document}

\title{L-systems for Measuring Repetitiveness\thanks{Funded in part by Basal Funds FB0001, Fondecyt Grant 1-200038, and a Conicyt Doctoral Scholarship, ANID, Chile.}}

\author{Gonzalo Navarro \and
    Cristian Urbina}
\authorrunning{Navarro and Urbina}

\institute{CeBiB --- Center for Biotechnology and Bioengineering \\ Departament of Computer Science, University of Chile}

\maketitle

\begin{abstract}
    An L-system (for lossless compression) is a CPD0L-system extended with two parameters $d$ and $n$, which determines unambiguously a string $w = \tau(\varphi^d(s))[1:n]$, where $\varphi$ is the morphism of the system, $s$ is its axiom, and $\tau$ is its coding. The length of the shortest description of an L-system generating $w$ is known as $\ell$, and is arguably a relevant measure of repetitiveness that builds on the self-similarities that arise in the sequence.
    
    In this paper we deepen the study of the measure $\ell$ and its relation with $\delta$, a better established lower bound that builds on substring complexity. Our results show that $\ell$ and $\delta$ are largely orthogonal, in the sense that one can be much larger than the other depending on the case. This suggests that both sources of repetitiveness are mostly unrelated. We also show that the recently introduced NU-systems, which combine the capabilities of L-systems with bidirectional macro-schemes, can be asymptotically strictly smaller than both mechanisms, which makes the size $\nu$ of the smallest NU-system the unique smallest reachable repetitiveness measure to date.
\keywords{L-systems \and Repetitiveness measures \and Text compression}
\end{abstract}

\section{Introduction}

In areas like Bioinformatics, it is often necessary to handle big collections of highly repetitive data. For example, two human genomes share $99.9\%$ of their content \cite{Przeworski2000}. In another scenario, for sequencing a genome one extracts so-called {\em reads} (short substrings) from it, with a ``coverage'' of up to 100X, which means that each position appears on average in 100 reads.\footnote{{\tt https://www.illumina.com/science/technology/next-generation-sequencing/ plan-experiments/coverage.html}} There is a need in science and industry to maintain those huge string collections in compressed form. Traditional compressors based exclusively on { Shannon's entropy} are not good for handling repetitive data, as they only exploit symbol frequencies when compressing. Finding good measures of repetitiveness, and also compressors exploiting this repetitiveness, has then become a relevant research problem.

A strong theoretical measure of string repetitiveness introduced by Kociumaka et al.\ is $\delta$ \cite{KNP20}, based on the substring complexity function. This measure has several nice properties: it is computable in linear time, monotone, resistant to string edits, insensitive to simple string transformations, and it lower-bounds almost every other theoretical or {\em ad-hoc} repetitiveness measure considered in the literature. Further, although $\delta$ is unreachable, there exist $O(\delta \log \frac{n}{\delta})$-space representations supporting efficient pattern matching queries, and this space is tight: no $o(\delta\log\frac{n}{\delta})$-space representation can exist \cite{KNP20}.

The idea that $\delta$ is a sound lower bound for the repetitiveness is reinforced by the fact that it is always $O(b)$, where $b$ is the size of the smallest {\em bidirectional macro-scheme} generating a string $w$ \cite{Storer1982}. Those macro-schemes arguably capture any possible way of exploiting copy-paste regularities in the sequences. Some very recent schemes \cite{NU2021}, however, explore other sources of repetitiveness, in particular {\em self-similarity}, and are shown to break the lower bound of $\delta$. 

The simplest of those schemes, called {\em L-systems} \cite{NU2021}, builds upon Lindenmayer systems \cite{Lindenmayer1968a,Lindenmayer1968b}, in particular on the variant called CPD0L-systems. A CPD0L-system describes the language of the images under a coding $\tau$, of the powers of a morphism $\varphi$ starting from an string $s$ (called the {\em axiom}), that is, the set $\{ \tau(\varphi^i(s))\,|\, i \ge 0 \}$. The L-system extends a CPD0L-system with two parameters $d$ and $n$, so as to unambiguously determine the string $w = \tau(\varphi^d(s))[1:n]$. The size of the shortest description of an L-system generating $w$ in this fashion is called $\ell$. Intuitively, $\ell$ works as a repetitiveness measure because any occurrence of the symbol $a$ at level $i$ expands to the same string at level $i+j$ for any $j$. 

Since $\ell$ is a reachable measure of repetitiveness (because the L-system is a representation of $w$ of size $\ell$), there are string families where $\delta = o(\ell)$. Intriguingly, it is shown \cite{NU2021} that there are other string families where $\ell = o(\delta)$, so (1) both measures are not comparable and (2) the lower bound $\delta$ does not capture this kind of repetitiveness. On the other hand, it is shown that $\ell = O(g)$, where $g$ is the size of the smallest deterministic context-free grammar generating only $w$. This comparison is relevant because L-systems are similar to grammars, differing in that they have no terminal symbols, so their expansion must be explicitly stopped at level $d$ and then possibly converted to terminals with $\tau$. 
Grammars provide an upper bound to repetitiveness that is associated with well-known compressors, so this upper bound makes $\ell$ a good measure of repetitiveness. 

A more complex scheme that was also introduced  \cite{NU2021} are NU-systems, which combine the power of L-systems with bidirectional macro-schemes. The measure $\nu$, defined as the size of the smallest NU-system generating $w$, naturally lower bounds both $\ell$ and $b$. The authors could not, however, find string families where $\nu$ is asymptotically better than both $\ell$ and $b$, so it was unclear if NU-systems are actually better than just the union of both underlying schemes.

In this paper we deepen the study of the relations between these new intriguing measures and more established ones like $\delta$ and $g$. Our results are as follows: 
\begin{enumerate}
\item  We show that $\ell$ can be much smaller than $\delta$, by up to a $\sqrt{n}$ factor, improving a previous result \cite{NU2021} and refuting their conjecture that $\ell = \Omega(\delta /\log n)$.
\item On the other hand, we expose string families where $\ell$ is larger than the output of several repetitiveness-aware compressors like the size $g_{rl}$ of the smallest run-length context-free grammar, the size $z_e$ of the smallest LZ-End parse \cite{Kreft2010}, and the number of runs $r$ in the Burrows-Wheeler Transform of the string \cite{BW94}. We then conclude that $\ell$ is incomparable to almost all measures other than $g$, which suggests that the source of repetitiveness it captures is largely orthogonal to the typical cut-and-paste of macro-schemes.
\item We introduce a string family where $\nu$ is asymptotically strictly smaller than both $\ell$ and $b$, which shows that NU-systems are indeed relevant and positions $\nu$ as the unique smallest reachable repetitiveness measure to date, capturing both kinds of repetitiveness in non-trivial ways.
\item We study various ways of simplifying L-systems, and show in all cases we end up with a weaker repetitiveness measure. We also study some of those weaker variants of $\ell$, which can be of independent interest.
\end{enumerate}

Overall, our results contribute to understanding how to measure repetitiveness and how to exploit it in order to build better compressors.

\section{Basic concepts}

In this section we explain the basic concepts needed to understand the rest of the paper, from strings and morphisms to relevant repetitiveness measures.

\subsection{Strings}
An {\em alphabet} is a finite set of {\em symbols}, and is usually denoted by $\Sigma$. A string $w$ is a sequence of symbols in $\Sigma$, and its length is denoted $|w|$. The {\em empty string}, whose length is $0$, is denoted by $\varepsilon$. The set of all possible finite strings over $\Sigma$ is denoted by $\Sigma^*$. The $i$-th symbol of $w$ is denoted by $w[i]$, if $1 \leq i \leq |w|$. The notation $w[i:j]$ stands for the subsequence $w[i]\dots w[j]$, if $1 \leq i \leq j \leq |w|$, or $\varepsilon$ otherwise. Other convenient notations are $w[:i] = w[1:i]$, and $w[j:] = w[j:|w|]$. If $x=x[1]\dots x[n]$ and $y=y[1] \dots y[m]$, the concatenation operation $x \cdot y$ (or just $xy$) stands for $x[1]\dots x[n]y[1]\dots y[m]$. Let $w = xyz$. Then $y$ is a {\em substring} (resp. $x$, $z$) of $w$ (resp. {\em prefix}, {\em suffix}), and it is {\em proper} if it is not equal to $w$.

A (right) {\em infinite string} \textbf{w} (we use boldface for them) over an alphabet $\Sigma$ is a mapping from $\mathbb{Z}^+$ to $\Sigma$, and its length is $\omega$, which is greater than any $n \in \mathbb{Z}^+$. The notations $w[i]$ and $w[i: j]$ carry over to infinite strings. It is possible to define the concatenation $x\cdot\textbf{y}$ if $x$ is finite and $\textbf{y}$ infinite. The concepts of substring, prefix and suffix also carry over, with non-trivial prefixes always being finite, and suffixes always being infinite. 

\subsection{Morphisms}

The set $\Sigma^*$ together with the (associative) concatenation operator and the (identity) string $\varepsilon$ form a {\em monoid} structure $(\Sigma^*, \cdot, \varepsilon)$. A {\em morphism} on strings is a function $\varphi: \Sigma_1^* \rightarrow \Sigma_2^*$ satisfying that $\varphi(\varepsilon)=\varepsilon$ and $\varphi(x \cdot y) = \varphi(x) \cdot \varphi(y)$ (i.e., a function preserving the monoid structure), where $\Sigma_1$ and $\Sigma_2$ are alphabets. To define a morphism of strings, it is sufficient to define how it acts over the symbols in its domain, which are called its {\em rules}, and there are $|\Sigma_1|$ of them. If $\Sigma_1 = \Sigma_2$, then the morphism is called an {\em endomorphism}.

Let $\varphi: \Sigma_1^* \rightarrow \Sigma_2^*$ be a morphism on strings. Some useful definitions are $width(\varphi) = max_{a \in \Sigma_1}|\varphi(a)|$ and $size(\varphi) = \sum_{a \in \Sigma_1}{|\varphi(a)|}$. A morphism is {\em non-erasing} if $\forall a \in \Sigma_1, |\varphi(a)| > 0$, {\em expanding} if $\forall a \in \Sigma_1, |\varphi(a)| > 1$, $k$-{\em uniform} if $\forall a \in \Sigma_1, |\varphi(a)| = k > 2$, and a {\em coding} if $\forall a \in \Sigma_1, |\varphi(a)| = 1$ (sometimes called a $1$-uniform morphism).

Let $\varphi: \Sigma^* \rightarrow \Sigma^*$ be an endomorphism. Then $\varphi$ is {\em prolongable} on a symbol $a$ if $\varphi(a) = ax$ for some string $x$, and $\varphi^i(x) \neq \varepsilon$ for every $i$. If this is the case, then for each $i,j$ with $0 \leq i \leq j$, it holds that $\varphi^i(a)$ is a prefix of $\varphi^j(a)$, and $\textbf{x} = \varphi^{\omega}(a) = ax\varphi(x)\varphi^2(x)\dots$ is the unique infinite fixed-point of $\varphi$ starting with the symbol $a$. An infinite string $\textbf{w} = \varphi^{\omega}(a)$ that is the fixed-point of a morphism is called a {\em purely morphic word}, its image under a coding $\textbf{x} = \tau(\textbf{w})$ is called a {\em morphic word}, and if the morphism $\varphi$ is $k$-uniform, then $\textbf{x}$ is said to be $k$-{\em automatic}. If $\textbf{w}$ is a purely morphic word, fixed-point of a morphism $\psi$, then there exist a coding $\tau$, a non-erasing morphism $\varphi$, and a symbol $a$ such that $\textbf{w} = \tau(\varphi^{\omega}(a))$ \cite{AlloucheShallit2003}. This also implies that we can generate any morphic word by iterating a non-erasing prolongable morphism, and then applying a coding.

\subsection{Repetitiveness measures}

\subsubsection{Grammars.}
A {\em straight line program} (SLP) is a deterministic context free grammar whose language is a singleton $\{w\}$. The measure $g$ is defined as the size of the smallest SLP $G$ generating ${w}$.
Finding the smallest SLP is an NP-complete problem \cite{Charikar2005}, although in practice, there exist algorithms providing log-approximations \cite{Jez2015,Rytter2003}. Another measure based on grammars is $g_{rl}$, the size of the smallest {\em run-length} SLP (RLSLP) generating $w$ \cite{Nishimoto2016}. RLSLPs allow constant-size rules of the form $A \rightarrow a^n$ for $n > 1$, which can make a noticeable difference in some string families like $\{\syma^n\, |\, n \geq 0\}$, where $g = \Theta(\log n)$, but $g_{rl}=O(1)$.

\subsubsection{Parsings.}

A {\em parsing} produces a factorization of a string $w$ into non-empty {\em phrases}, $w = w_1w_2\dots w_k$. Several compressors work by parsing $w$ in a way that storing summary information about the phrases enables recovering $w$. 

The {\em Lempel-Ziv} parsing (LZ) process a string from left to right, always forming the longest phrase that has a copy starting inside some previous phrase \cite{LZ76}. The LZ-no parsing always forms the longest phrase with a copy fully contained in the concatenation of previous phrases. Another variation is the LZ-end parsing, which forms the longest phrase with an occurrence ending in alignment with a previous phrase \cite{Kreft2010}. All of these parsings can be constructed in linear time, and their number of phrases are denoted by $z$, $z_{no}$, and $z_e$, respectively.

A {\em bidirectional macro-scheme} (BMS) \cite{Storer1982} is any parsing where each phrase of length greater than $1$ has a copy starting at a different position, in such a way that the original string can be recovered following these pointers. The measure $b$ is defined as the size of the smallest BMS for $w$. It strictly lower bounds asymptotically all the other reachable repetitiveness measures \cite{NOP20}, except for $\ell$ and $\nu$ \cite{NU2021}. It is NP-hard to compute \cite{Gallant1982}, though.

\subsubsection{Burrows-Wheeler transform.}

The Burrows-Wheeler transform (BWT) \cite{BW1994} is a reversible transformation that usually makes a string more compressible. It is obtained by concatenating the last symbols of the sorted rotations of $w$. The BWT tends to produce long runs of the same symbol when a string is repetitive, and these runs can be compressed into one symbol using {\em run-length encoding} ($rle$). A repetitiveness measure based in this idea is defined as $r(w) =|rle(BWT(w))|$. Although $r$ is not so good as a repetitiveness measure \cite{Giuliani2021}, it has practical applications representing sequences in Bioinformatics \cite{GNP18}.

\subsubsection{Substring complexity.}

Let $F_w(k)$ be the set of substrings of $w$ of length $k$. The {\em complexity function} of $w$ is defined as $P_w(k) = |F_w(k)|$. Kociumaka et al. introduced a repetitiveness measure based on the complexity function, defined as $\delta(w) = max\{P_w(k)/k\,|\,k\in[1..|w|]\}$ \cite{KNP20}.
This measure has several nice properties: it is computable in linear time, monotone, insensitive to reversals, resistant to small edits on $w$, can be used to construct $O(\delta \log \frac{n}{\delta})$-space representations supporting efficient access and pattern matching queries \cite{KNP20}, and is a lower bound to almost every other theoretical or {\em ad-hoc} repetitiveness measure considered in the literature. On the other hand, $o(\delta\log\frac{n}{\delta})$ space has been proved to be unreachable \cite{KNP20}.

\section{The measure $\ell$}

The class of {\em CPD0L-systems} is a variant of the original {\em L-systems}, parallel grammars without terminals, defined by Aristid Lindenmayer to model cell divisions in the growth of plants and algaes \cite{Lindenmayer1968a,Lindenmayer1968b}. 

A CPD0L-system is a 4-tuple $L=(\Sigma, \varphi, \tau, s)$, where $\Sigma$ is the {\em alphabet}, $\varphi$ is the set of {\em rules} (an endomorphism on $\Sigma^*$), $\tau$ is a coding on $\Sigma^*$, and $s \in \Sigma^+$ is the {\em axiom}. The system generates the language $\{\tau(\varphi^d(w))\, |\, d \geq 0\}$.  The ``D0L'' stands for {\em deterministic L-system with 0 interactions}. The ``P'' stands for {\em propagating}, which means that it has no $\varepsilon$-rules. The ``C'' stands for {\em coding}, which means that the system is extended with a coding. 
For a CPD0L-system to be utilizable as a compressor, we extend it to a 6-tuple with two extra parameters, $d$ and $n$, and define the unique string generated by the system as $\tau(\varphi^d(w))[1:n]$. For simplicity, in the rest of this paper, we refer to these extended CPD0L-systems as L-systems.

The measure $\ell$ is defined as the size of the smallest L-system generating a string, where the {\em size} of the L-system is $size(\varphi)+|s|+|\Sigma|+2$, accounting for the lengths of the right-hand sides of its rules, the length of the axiom, the function $\tau$, and the values $d$ and $n$. This space is measured in $O(\log n)$-bit words, so we always assume that $d = n^{O(1)}$ and that $\Sigma = n^{O(1)}$.
A finer-grained analysis about the number of bits needed to represent an  L-system of size $\ell$ yields $O(\ell \log |\Sigma| + \log n)$ bits, the second term corresponding to $d$ and $n$; note that  $\Sigma$ contains the alphabet of $w$.

An important result about $\ell$ is that it always holds that $\ell = O(g)$ \cite{NU2021}. More importantly, sometimes $\ell = o(\delta)$, which implies that $\delta$ is not lower bound for $\ell$, and questions $\delta$ as a golden measure of repetitiveness. 

\subsection{Variants}

To understand the particularities of $\ell$, we study several classes of L-systems with different restrictions, and define measures based on them. We define the measure $\ell_e$ (resp., $\ell_u$) that restricts the morphism to be expanding (resp., $k$-uniform). The variant $\ell_m$ forces the morphism to be $a$-prolongable for some symbol $a$, and the axiom to be $s = a$. The variant $\ell_d$ essentially removes the coding. Finally, $\ell_p$ refers to the intersection of $\ell_m$ and $\ell_d$, and $\ell_a$ refers to the intersection of $\ell_m$ and $\ell_u$.

\begin{definition}
    An L-system $(\Sigma, \varphi, \tau, s)$ is {\em $a$-prolongable} if there exists a symbol $a$ such that $s = a$ and $a \rightarrow ax$ with $x \neq \varepsilon$. An L-system is {\em prolongable} if it is $a$-prolongable for some symbol $a$.
\end{definition}

\begin{definition}The $\ell$-variants studied in this paper are the following:
\begin{itemize}
\item The measure $\ell$ denotes the size of the smallest L-system generating $w$.
\item The variant $\ell_e$ denotes the size of the smallest L-system generating $w$, satisfying that all its rules have size at least 2.
\item The variant $\ell_m$ denotes the size of the smallest prolongable L-system generating $w$.
\item The variant $\ell_d$ denotes the size of the smallest L-system generating $w$, satisfying that $\tau$ is the identity function.
\item The variant $\ell_u$ denotes the size of the smallest L-system generating $w$, satisfying that all its rules have the same size, at least 2.
\item The variant $\ell_p$ denotes the size of the smallest prolongable L-system generating $w$, satisfying that $\tau$ is the identity function.
\item The variant $\ell_a$ denotes the size of the smallest prolongable L-system generating $w$, satisfying that all its rules have the same size, at least 2.
\end{itemize}
\end{definition}

Our results concerning the proposed $\ell$-variants and other relevant repetitiveness measures across the paper are summarized in Figure~\ref{fig:variants}. 

\begin{figure}[t]\centering
\begin{tikzpicture}
[>=latex, 
node distance=0.5cm and 2.0cm,
rounded corners=0.25cm,
main/.style = {draw, fill=gray!10, minimum size=0.6cm}]
\node[main] (ell) {$\ell$};
\node[main] (elle) [right=of ell] {$\ell_e$};
\node[main] (ellu) [right=of elle]{$\ell_u$};
\node[main] (ella) [right=of ellu]{$\ell_a$};
\node[main] (ellm) [above=of elle] {$\ell_m$};
\node[main] (ellp) [above=of ella] {$\ell_p$};
\node[main] (elld) [above=of ellm]{$\ell_d$};
\node[main] (g) [right=of elld]{$g$};
\node[main] (nu) [left=of ell]{$\nu$};
\node[main] (b) [above=of ell]{$b$};
\node[main] (delta) [above=of nu]{$\delta$};
\coordinate[above=of g] (gspace);
\coordinate[above=of gspace] (gspace2);
\coordinate[right=2.625cm of gspace2] (gspace3);

\draw[->] (nu) to (ell);
\draw[->] (ell) to (elle);
\draw[->] (ell) to (ellm);
\draw[->, dashed] (ell) to (elld);
\draw[->, dashed] (elle) to (ellu);
\draw[->, dashed] (ellu) to (ella);
\draw[->] (ellm) to (ellp);
\draw[->, dashed] (ellm) to (ella);
\draw[->] (elld) to (g);
\draw[->] (elld) to (ellp);
\draw[->, gray] (elld) to (ellm);
\draw[->, gray] (ellu) to (ellm);

\draw[->] (nu) to (b);
\draw[->] (delta) to (b);
\draw[->] (b) |- (gspace) -- (g);

\draw[<->, gray] (delta) |-  (gspace3) -- (ellp);
\draw[<->, gray] (b) to (ell);
\draw[<->, gray] (nu) to (delta);
\draw[<->, gray] (g) to (ellp);

\end{tikzpicture}\caption{Asymptotic relations between $\ell$-variants and other relevant measures. A black arrow (dashed or solid) from $v_1$ to $v_2$ means that it holds that $v_1 = O(v_2)$ for any string family. If the black arrow is solid, then also there exists a string family where $v_1 = o(v_2)$. A gray arrow from $v_1$ to $v_2$ means that we known that there exists a family where $v_1 = o(v_2)$.}\label{fig:variants}
\end{figure}
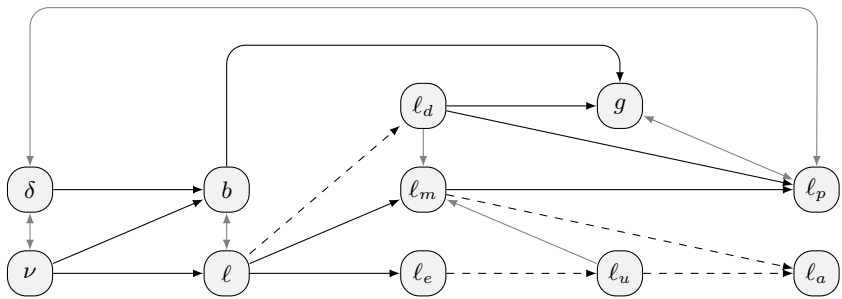

\section{A family where $\ell$ is much better than  $\delta$}

Navarro and Urbina showed a string family satisfying that $\delta = \Omega(\ell \log n)$ \cite{NU2021}, and conjectured that this gap was the maximum possible, that is, that the lower bound $\ell = \Omega(\delta / \log n)$ holds for any string family. We now disprove this conjecture. We show a string family where $\delta$ is $\Theta(\sqrt{n})$ times bigger than the size $\ell$ of the smallest  L-system.

\begin{lemma}\label{thm:ellp<delta}
    There exists a string family where $\delta = \Theta(\ell \sqrt{n})$.
\end{lemma}

\begin{proof}
    Consider the \symc-prolongable L-system defined as $$(\{\syma,\symb,\symc\}, \{\syma \rightarrow \syma, \symb \rightarrow \syma\symb,\symc \rightarrow \symc\symb\}, id, \symc, d+1, 1 + \frac{d(d+1)}{2} + d)$$
    for any $d$. Clearly this system generates strings of the form $s_d = \symc\Pi_{i=1}^d\syma^i\symb $ for $d \geq 0$. It holds that $\ell$ is $O(1)$ in this family.
    On the other hand, the first $1+ (d/2)(d/2+1)/2 + d/2$ substrings of length $d$ of $\symc\Pi_{i=1}^d\syma^i\symb $ are completely determined by the $\symb$'s they cross, and the number of $\syma$'s at their extremes, so they are all distinct. This gives the lower bound $\delta = \Omega(d) = \Omega(\sqrt{n})$, and the upper bound $O(\sqrt{n})$ holds for run-length grammars, so $\delta = \Theta(\sqrt{n})$. Thus $\delta = \Theta(\ell \sqrt{n})$ in this string family.\qed
\end{proof}

It is curious how this string family is so easy to describe, yet so hard to represent with any copy-paste mechanism. Intuitively, the simplicity of the sequence depends in that any factor is almost the same as the previous one, so it is arguably highly repetitive, just not via copy-paste. As we prove in Section \ref{sec:variants}, the variant $\ell_p$ is in general, pretty bad, so $\delta$ being incomparable to this weak variant (the system in the proof is prolongable and has identity coding) is even more surprising. 

\section{Incomparability of $\ell$ with other measures}\label{subsection:incomparability}

It is known that $\ell = O(g)$\cite{NU2021} (their proof applies to $\ell_d$ as well), which shows that the measure $\ell$ is always reasonable for repetitive strings. But as other reachable measures, $\ell$ has its own drawbacks. We prove that in general, it does not hold that $\ell = O(g_{rl})$, making L-systems incomparable to RLSLPs \cite{Nishimoto2016}. 

\begin{lemma}\label{thm:ell_grl}
    There exists a string family where $\ell = \Omega(g_{rl}\log n/\log\log n)$.
\end{lemma}

\begin{proof}
    Kociumaka et al. showed a string family needing $\Omega(\log^2 n)$ bits to be represented with any method \cite{KNP20}. This family is composed of all the strings that can be constructed by extracting a prefix of the characteristic sequence of the powers of 2, and then for every $k$-th symbol $\one$ in this prefix, moving it forward up to $2^{k-1}$ positions. Strings constructed in this form have $O(\log n)$ runs of \zero's separated by \one's, so $g_{rl} = \Theta(\log n)$ in this family. The minimal L-system for a string in this family can be represented with $O(\ell \log |\Sigma| + \log n) \subseteq O(\ell\log\ell + \log n)$ bits, and this must be $\Omega(\log^2 n)$, so it follows that $\ell = \Omega(\log^2 n/\log \log n)$. Thus, $\ell = \Omega(g_{rl}\log n/\log\log n)$ in this string family.\qed
\end{proof}

The same result holds for LZ parsings \cite{LZ76}. Even the LZ-End parsing \cite{Kreft2010} (the biggest of them) can be asymptotically smaller than $\ell$ in some string families.

\begin{lemma}
    There exists a string family where $\ell = \Omega(z_e\log n/\log\log n)$.
\end{lemma}

\begin{proof}
    Take each string of length $n$ in the family of the proof of Lemma \ref{thm:ell_grl}, and preppend $\zero^n$ to it. This new family of strings still needs $\Omega(\log^2 n)$ bits to be represented with any method, because their amount is the same and $n$ just doubled. Just as before, then, $\ell = \Omega(\log^2 n/\log \log n)$ in this family. On the other hand, the LZ-End parsing needs $\Theta(\log n)$ phrases only to represent the prefix $\zero^n\one$, and then for each run of \zero's followed by \one, its source is aligned with $\zero^n\one$, so $z_e = \Theta(\log n)$. Thus, $\ell = \Omega(z_e\log n/\log\log n)$.\qed
\end{proof}

The same result also holds for the number of runs in the Burrows-Wheeler transform \cite{BW94} of a string.

\begin{lemma}
    There exists a string family where $\ell = \Omega(r\log n/\log\log n)$.
\end{lemma}

\begin{proof}
   Consider the family of the proof of Lemma \ref{thm:ell_grl}. Clearly $r = \Theta(\log n)$, because $r$ is reachable, and the BWT also has at most $O(\log n)$ runs of \zero's separated by \one's. Thus, $\ell = \Omega(r\log n/\log\log n)$ in this string family.\qed
\end{proof}

We conclude that the measure $\ell$ is incomparable to almost every other repetitiveness measure. We summarize this in the following theorem (see \cite{Navarro2021} for the measures not explained in this paper).

\begin{theorem}
    The measure $\ell$ is incomparable with the repetitiveness measures $\delta, \gamma, b, v, c, g_{rl}, z, z_{no}, z_e$ and $r$. On the other hand, it holds that $\ell = O(g)$ and $\ell = \Omega(\nu)$ \cite{NU2021}.
\end{theorem}

\section{NU-systems and the measure  $\nu$}

A {\em NU-system} is a tuple $N = (V, R, \tau, s, d, n)$ that generates a unique string in a similar way to an L-system. The key difference is that in the right-hand side of its rules, a NU-system is permitted to have special symbols of the form $a(k)[i:j]$, whose meaning is to generate the $k$-th level from $a$, then extract the substring starting at position $i$ and ending at position $j$, and finally apply the coding to the resulting substring. The indexes inside a NU-system (e.g., levels, intervals) must be of size less or equal to $n$ to fit in an $O(\log n)$ bits word. Also, the NU-system must not produce any loops when extracting a prefix from some level, which is decidable to detect. The size of a NU-system is defined analogously to the size of L-systems, with the extraction symbols $a(k)[i:j]$ being symbols of length $4$. The measure $\nu$ is defined as the size of the smallest NU-system generating a string $w$, and it holds that $\nu = O(\ell)$ and $\nu = O(b)$ \cite{NU2021}. There exist families where both asymptotic bounds are strict.

We now show that NU-systems exploit the features of L-systems and macro-schemes in a way that, for some string families, can reach sizes that are unreachable for both L-systems and macro-schemes independently.

\begin{theorem}\label{thm:nu}
    There exists a family of strings where $\nu = o(\min(\ell, b))$.
\end{theorem}

\begin{proof}
    Let $\mathcal{F}$ be the family of strings defined by Kociumaka et al., needing $\Omega(\log^2 n)$ bits to be represented with any method \cite{KNP20}. We construct a new family $\mathcal{F}' = \{x \cdot \textbf{y}[:n]\,|\, x \in \mathcal{F} \land |x| = n\}$, where \textbf{y} is the infinite fixed point generated by the \symc-prolongable L-system with identity coding utilized in Lemma \ref{thm:ellp<delta}.

    It still holds that $\ell = \Omega(\log^2 n / \log \log n)$ in this family. On the other hand, $b = \Omega(\sqrt{n})$, because $\delta = \Omega(\sqrt{n})$ on prefixes of \textbf{y}, and the alphabets between the prefix in $\mathcal{F}$ and $\textbf{y}[:n]$ are disjoint.
    
    Let $x$ be a string in $\mathcal{F}$ of length $n$, with $k$ symbol $\one$'s. 
    Let $i_{j}$ be the number of \zero's between the $(j-1)$-th symbol $\one$, and the $j$-th symbol \one, for $j \in [2, k$], in $x$. Also, let $i_1$ and $i_{k+1}$ be the number of \zero's at the left and right extremes of $x$.
    We then construct a NU-system, where $\tau$ is the identity coding, $d = 1$, and the prefix length is $2n$:
    \begin{align*}
    V &= \{\zero, \one, \syma, \symb, \symc\}\\
    R &= \{\zero \rightarrow \zero\zero, \one \rightarrow \one, \syma \rightarrow \syma, \symb \rightarrow \syma\symb,\symc \rightarrow \symc\symb\} \\
    S &= \zero(n)[:i_1]\one\zero(n)[:i_2]\one \dots \zero(n)[:i_k] \one \zero(n)[:i_{k+1}]\symc(n)[:n]
    \end{align*}
    
    By construction, this NU-system generates the string $x\cdot\textbf{y}[:n]$ of length $2n$, and has size $4(k+2) + k + 8$. Thus, $\nu$ is $O(\log n)$ for these strings, and $\nu = o(\min(\ell, b))$ in $\mathcal{F}'$. \qed    
\end{proof}

NU-systems can then be smaller representations than those produced by any other compression method exploiting repetitiveness. On the other hand, though computable, no efficient decompression scheme has been devised for them.

\section{$\ell$-variants are weaker than $\ell$}\label{sec:variants}

We start this section by showing that $\ell$ can be asymptotically strictly smaller than $\ell_m$, that is, restricting L-systems to be prolongable has a negative impact in its compression ability.

\begin{lemma} There exists a string family where $\ell = o(\ell_m)$.
\end{lemma}

\begin{proof}Let $\mathcal{F} = \{\zero^n\one\, |\, n \geq 0\}$. It is clear that $\ell$ is constant in this family: the L-system $(\{\zero,\one\}, \{\zero \rightarrow \zero, \one \rightarrow \zero\one\}, id, \one, n, n)$ produces each string in $\mathcal{F}$ with the corresponding value of $n$.
    
For the sake of contradiction, suppose that $\ell_m = O(1)$ in $\f$. Let $L_n = (\Sigma_n, \varphi_n, \tau_n, a, d_n, n)$ be the the smallest $a$-prolongable morphism generating $\zero^n\one$. Because $\ell_m = O(1)$, there exists a constant $C$ satisfying that $|\Sigma_n| < C$  and $width(\varphi_n) < C$ for every $n$. Observe that it is only necessary to have one symbol $b$ with $\tau_n(b) = \one$, so w.l.o.g. assume that $\tau_n(\one) = \one$. As the system is $a$-prolongable, each level is a prefix of the next one. This implies that the morphism should be iterated until $\one$ appears for the first time, and then we can extract the prefix. This must happen in the first $C$ iterations of the morphism, otherwise $\one$ is not reachable from $a$ (i.e., if an iteration does not yield a new symbol, then no new symbols will appear since then, and there are at most $C$ symbols). But in the first $C$ iterations we cannot produce a string longer than the constant $C^C$, and there exists a finite number of strings of length less than $C^C$. For sufficiently large $n$, this implies that the symbol $\one$, if it is reachable, will appear for the first time before the $(n+1)$-th symbol, which is a contradiction.\qed
\end{proof}

Clearly, it also holds that $\ell_d = o(\ell_m$) in this family. Similarly, it is not difficult to see that $\ell_u$ is constant in the family $\{\zero^{2^n}1\, |\, n \geq 0\}$ (e.g., axiom $s = \texttt{01}$ and rules $\zero \rightarrow \zero\zero$, $\zero \rightarrow \one\one$). A similar argument yields that $\ell_u = o(\ell_m)$ for this other string family. 

Further, we can find a concrete asymptotic gap between $\ell$ and $\ell_m$ in the string family of the proof of the previous lemma.

\begin{lemma}
There exists a string family  where $\ell_m = \Omega(\ell \log n /\log \log n)$.
\end{lemma}

\begin{proof} Let $\mathcal{F} = \{\zero^n\one\, |\, n \geq 0\}$. Recall that $\ell = O(1)$ in this family. Let $k= |\Sigma|$ and $t = width(\varphi)$ obtained from the morphism of the smallest $a$-prolongable system generating $\zero^n\one$ (we assume again that the only symbol mapped to $\one$ by the coding is \one). In the first $k$ iterations, $\one$ must appear (as in the previous proof) and cannot be deleted in the following levels, so it cannot appear before position $n+1$. Hence, $t^k > n$, which implies $k > \log_t n$. By definition, $\ell_m \geq k \geq \log_t n$ and $\ell_m \geq t$, so $\ell_m \geq \max(t, \log_t n)$. The solution to the equation $t=\log_t n$ is the smallest value that $\max(t, \log_t n)$ can take for $t \in [2..n]$. This value is $\Omega(\log n/W(\log n))$ where $W(x)$ is the Lambert W function, and it holds that $W(\log n) = \Theta(\log \log n)$. Therefore, $\ell_m = \Omega(\ell \log n /\log \log n)$ in this string family. \qed
\end{proof}

We now show that, if we remove the coding from prolongable L-systems, we end with a much worse measure. 

\begin{lemma}
 There exists a string family where $\ell_p = \Omega(\ell_m \sqrt{n})$.
\end{lemma}

\begin{proof}
We prove that $\ell_p = \Theta(n)$ whereas $\ell_m = O(\sqrt{n})$ on $\mathcal{F} = \{\zero^n\one\, |\, n \geq 0\}$. Any prolongable morphism with an identity coding generating $\zero^n\one$ must have the rule $\zero \rightarrow \zero^n\one$, which implies $\ell_p = \Theta(n)$.
The reason is that if  the system is prolongable, but it has no coding, then the axiom must be \zero, and in the prolongable rule $\zero \rightarrow \zero w$, if $|\varphi(\zero)| \leq n+1$, then the non-empty string $w$ could only contain \zero's and \one's, otherwise undesired symbols would appear in the final string because the starting level is a prefix of the final level.
If $w$ does not contains \one's, then $\one$ is unreachable from \zero.
If $w$ contains a $\one$, then the first of them should be
at position $n+1$.

On the other hand, we can construct an $\syma$-prolongable morphism, with $\tau(\one) = \one$ and $\tau(a) = \zero$ for every other symbol $a \neq \one$ as follows: Let $n = k\lfloor \sqrt{n}\rfloor+j$ with $\sqrt{n} > 2, k > 1$, $j \geq 0$ ($k$ and $j$ integers), and define the following rules
\begin{align*}
\syma &\rightarrow \syma\symb \\
\symb &\rightarrow \symc^{k-1}\symd\\
\symc &\rightarrow \zero^{\sqrt{n}-1} \\
\symd &\rightarrow \zero^{\sqrt{n}-3+j}\one
\end{align*}
It holds that $\varphi^3(\syma) =\texttt{abc}^{k-1}\symd \zero^{(\sqrt{n}-1)(k-1)}\zero^{\sqrt{n}-3 + j}\one$, and that
$$|\varphi^3(\syma)| = 3 + (k-1) + (\sqrt{n}-1)(k-1) + (\sqrt{n} - 3 + j) + 1 =  n+1$$
Thus $\tau(\varphi^3(\syma)) = \zero^n\one$ as required, and the size of the system is clearly $O(\sqrt{n})$. \qed
\end{proof}

Surprisingly, this weak measure $\ell_p$ can be much smaller than $\delta$ for some string families. This can be deduced from Lemma \ref{thm:ellp<delta}. On the other hand, it does not hold that $\ell_p = O(g)$ for any string family, because $g = \Theta(\log n)$ in $\{\zero^n\one\, |\, n \geq 0\}$.

If we restrict L-systems to be expanding, that is, with all its rules of length at least 2, we also end with a weaker measure. This shows that in general, it is not possible to transform L-systems into expanding ones without incurring an increase in size.

\begin{lemma}
There exists a string family where $\ell = o(\ell_e)$.
\end{lemma}

\begin{proof}
Let $\mathcal{F} = \{\zero^n\one\zero^{2^n}\, |\, n \geq 0\}$. Clearly $\ell$ is constant in $\mathcal{F}$: the L-system $(\{\zero,\one, \two\}, \{\zero \rightarrow \zero\zero, \one \rightarrow \two\one, \two \rightarrow \two, \{\zero \rightarrow \zero, \one \rightarrow \one,  \two \rightarrow \zero\}, \one\zero, n, 2^n+n+1)$ produces $\zero^n\one\zero^{2^n}$ and stays constant-size as $n$ grows.

Suppose that $\ell_e$ is also constant in $\mathcal{F}$. Then there is a constant $C$ such that the minimal expanding L-systems generating the strings in this family have at most $C$ rules, each of length at most $C$. Without loss of generality, assume that for each of these systems, the only symbol mapped to $\one$ by the coding  is $\one$, also assume that the axiom is a single symbol $a_0$. Note that because the systems are expanding with rules of size at most $C$, their level must be $d \geq \log_C 2^n = \frac{n}{\log_2 C}$. 
Let $a_0, a_1, \dots, a_d$ be the sequence of first symbols of $\varphi^i(a_0)$ for $i \leq d$. By the pigeonhole principle, for sufficiently big values of $n$, this sequence has a period of length $q$ starting from $a_p$, with $p+q \leq C \leq d$. Then there exist indexes $t$ and $j$ such that $t = d - jq$ and $p \leq t <  p + q$. By the $q$-periodicity of the sequence starting at $a_t$, it is clear that $\varphi^q(a_t) = a_tw$ for some $w \neq \varepsilon$ (because the morphism is expanding), so $\varphi^q$ is prolongable on $a_t$. This implies that $\varphi^{iq}(a_t)$ is a prefix of $\varphi^{jq}(a_t)$ for $i \leq j$. As before, if $\one$ is reachable from $a_t$ via $\varphi^q$, that must happen in the first $C$ iterations, so $\varphi^{Cq}(a_t)$ contains a $\one$ symbol, and so does $\varphi^{jq}(a_t)$, which is a prefix of $\varphi^d(a_0)$. This implies that $\varphi^d(a_0)$ contains a $\one$ before position $C^{Cq}$, which is bounded by $C^{C^2}$, a contradiction for sufficiently long strings in the family. So it has to be that $\one$ is not reachable via $\varphi^q$ from $a_t$, but this is also a contradiction for sufficiently long strings, because $\varphi^{jq}(a_t)$ is a prefix of $\varphi^d(a_0)$ of length at least $2^{d-t} = \omega(n)$, yielding too many symbols not mapped to $\one$ before the first \one at level $d$. Thus $\ell_e$ cannot be $O(1)$ in $\mathcal{F}$.
\qed
\end{proof}

\section{Conclusions and open questions}

The measure $\ell$ is arguably a strong reachable measure, which can break the limits of $\delta$ (a measure considered a stable lower bound for repetitiveness) by a wide margin (i.e., a factor of $\sqrt{n}$). On the other hand, however, $\ell$ can be asymptotically weaker than the space reached by several compressors based on run-length context-free grammars, many Lempel-Ziv variants, and the Burrows-Wheeler Transform. Only the size of context-free grammars provides an upper bound to $\ell$. This suggests that the self-similarity exploited by L-systems is mostly independent of the  source of repetitiveness exploited by other  compressors and measures, which build on copy-paste mechanisms. We also show that the definition of L-systems is robust, in the sense that several attempts to simplify or restrict them leads to weaker measures.

A relevant question about L-systems is whether they can be useful to build compressed sequence representations that support direct access. More formally, can we build an $O(\ell)$-space representation of a string $w[1:n]$ providing random access to any symbol in $O(\text{polylog}\, n)$ time? 
The closest result (as far as we know) is an algorithm designed by Shallit and Swart \cite{ShallitSwart1999}, which computes $\varphi^d(a)[i]$ in time bounded by a polynomial in $|\Sigma|, width(\varphi), \log d$ and $\log i$. It uses more space and takes more time than our aim. The main bottleneck is having to store the incidence matrix of the morphism and to calculate its powers. As suggested by Shallit and Swart, this could be solved by finding closed forms for the growth functions (recurrences) of each symbol. If this approach were taken, these formulas should be easy to describe within $O(\ell)$ space.

We leave some other conjectures about $\ell$. The first is that $\ell = \Omega(\delta/\sqrt{n})$ for any string family. The second is that it holds that $\ell = o(\ell_d)$ for some string family. Another interesting task is identifying whose of the $\ell$-variants are $O(g)$. 

In terms of improving compression, on the other hand,
the recent measure $\nu$ \cite{NU2021} aims to unify the repetitiveness induced by self-similarity and by explicit copies. This measure is the smallest size of a NU-system, a natural way to combine L-systems (with minimum size $\ell$) with macro-schemes (with minimum size $b \ge \delta$). In line with our finding that $\ell$ and $\delta$ are mostly orthogonal, we prove in this paper that $\nu$ is strictly more powerful than both $\ell$ and $b$, which makes $\nu$ the unique smallest reachable measure of repetitiveness to date. 

Because the aim for $\ell$ was also defining a practical measure, $\varepsilon$-rules were not allowed in the original definition. It would still be a computable measure if we relax that restriction. Another feature that could increase L-systems power is the use of a morphism instead of a coding (this is common the in literature on L-systems), or adding run-length rules. We leave the study of these features for the future, as we are already working on understanding $\nu$.

There are several open questions related to NU-systems and $\nu$. For example, does it hold that $\nu = \Omega(\ell \log \log n/ \log n)$, or $\nu = \Omega(\delta/\sqrt{n})$, for every string family?
Is $\nu = O(\gamma)$, or at least $o(\gamma \log (n/\gamma))$, for every string family? (recall that $\gamma$ is a measure between $\delta$ and $b$, and $o(\gamma \log (n/\gamma))$ is unknown to be reachable \cite{Kempa2018}).
And towards having a practical compressor based on $\nu$, can we decompress a NU-system efficiently?

In a more general perspective, 
this paper pushes a little further the discussion of what we understand by a repetitive string. Intuitively, repetitiveness is about copies, and macro-schemes capture those copies pretty well, but there are many other aspects in a text that could be repeating besides explicit copies, such as general patterns and relative ordering of symbols. Macro-schemes capture explicit copies, L-systems capture self-similarity, and NU-systems capture both. What other regularities could we exploit when compressing strings, keeping the representation (more or less) simple and the associated repetitiveness measure (hopefully efficiently) computable? 
%
%
%
\bibliographystyle{splncs04}
\bibliography{bibliography}

\end{document}